\documentclass[10pt, conference, compsocconf]{IEEEtran}

\usepackage{graphics}
\usepackage{graphicx}

\usepackage{epsfig}
\usepackage{color}
\usepackage{amsmath}
\usepackage{amsfonts}
\usepackage{amsthm}
\usepackage{latexsym}
\usepackage[tight,footnotesize]{subfigure}
\usepackage{balance}
\usepackage{tikz}
\usetikzlibrary{shapes.geometric}

\newtheorem{theorem}{{\bf Theorem}}
\newtheorem{corollary}{{\bf Corollary}}

\newtheorem{cond}{Condition}

\newcommand{\comment}[1]{}

\newcommand{\IGNORE}[1]{}

\title{Free Energy Approximations for CSMA networks}

\author{B. Van Houdt\\
Dept. Mathematics and Computer Science\\
University of Antwerp, Belgium}
\date{}

\begin{document}

\maketitle
\begin{abstract}

In this paper we study how to estimate the back-off rates 
in an idealized CSMA network consisting of $n$ links
to achieve a given throughput vector using free energy
approximations. 
More specifically, we introduce the class of
region-based free energy approximations with clique belief and present a closed form
expression for the back-off rates based on the zero gradient points of the free energy
approximation (in terms of the conflict graph, target throughput vector and counting numbers).

Next we introduce the size $k_{max}$ clique free energy approximation as a special case and
derive an explicit expression for the counting numbers, as well as a recursion to
compute the back-off rates. We subsequently show
that the size $k_{max}$ clique approximation coincides with a Kikuchi free energy approximation 
and prove that it is exact on chordal conflict graphs when $k_{max} = n$. As a by-product these results provide us with an
explicit expression of a fixed point of the inverse generalized belief 
propagation algorithm for CSMA networks.

Using numerical experiments we compare the accuracy of the novel approximation method with existing methods.
\end{abstract}

\section{Introduction}
Carrier sense multiple access (CSMA) networks form an attractive random access solution for wireless networks
due to their fully distributed nature and low complexity. In order to guarantee a certain set of
feasible throughputs  for the links part of a CSMA network (defined as the fraction of the time that a
link is active), the {\it back-off rate} of each link has to be set
in the appropriate manner which depends on the
network topology, i.e., how the different links in the network interfere with each other. 
In order to obtain a better understanding of how these back-off rates affect the throughputs,
the ideal CSMA network model was introduced (see \cite{boorstyn1,durvy2,jiang2,jiang1,vandeven5,vandeven1,yun1})
and this model was shown to provide good estimates for the throughput achieved 
in real CSMA like networks \cite{wang5}.

The product form solution of the ideal CSMA model was established long ago \cite{boorstyn1} 
(for exponential back-off durations) and the set $\Gamma$ of achievable throughput vectors $\vec \phi = (\phi_1,\ldots,\phi_n)$,
where $\phi_i$ is the throughput of link $i$,  was characterized in \cite{jiang1}.
Further, for each vector $\vec \phi \in \Gamma$ the existence of a unique vector of back-off rates that achieves $\vec \phi$ was proven in \cite{vandeven5}.
None of the above results indicates how to set the back-off rates to achieve a given vector $\vec \phi \in \Gamma$
(except for very small networks).
For line networks with a fixed interference range this problem was solved in \cite{vandeven1} for any
$\vec \phi = (\alpha,\ldots,\alpha) \in \Gamma$,
while \cite{yun1} presented a closed form expression for the back-off rates to achieve any 
$\vec \phi \in \Gamma$ in case the conflict graph is a tree. This expression was obtained from the
zero gradient points of the {\it Bethe free energy} and can be used as an approximation in general conflict graphs,
termed the Bethe approximation. 
The explicit results for line and tree networks were generalized in \cite{vanhoudt_ton17}, where a simple formula for
the back-off rates was presented for any {\it chordal} conflict graph $G$. This formula was subsequently
used to develop the local chordal subgraph (LCS) approximation for general conflict graphs.  

Another method to approximate unique back-off rates given an achievable target throughput vector 
exists in using inverse (generalized) belief
propagation (I(G)BP) algorithms \cite{kai1,yedidia2}.  These algorithms are 
message passing algorithms that in general are not guaranteed to converge to a fixed point. 
In \cite{kai1}  the IBP algorithm for CSMA was argued to converge to the exact vector of back-off rates
when the conflict graph is a tree, but convergence of IGBP for loopy graphs to a (unique) fixed point was not established. 
Belief propagation algorithms are intimately related to free energy approximations as their fixed points can be shown 
to correspond to the zero gradient points of an associated free energy approximation \cite{yedidia1}.

The main objective of this paper is to introduce more refined free energy approximations (compared to the Bethe approximation)
for the ideal CSMA model that yield closed form approximations for the back-off rates and to compare their accuracy with the Bethe and LCS approximation.
The contributions of the paper are as follows. 

First, we introduce a class of region-based free energy approximations with clique belief and 
a closed form expression for the CSMA back-off rates based on its zero gradient points (Section \ref{sec:cliquebelief}).

Second, we propose the size $k_{max}$ clique approximation as a special case and present a closed form
expression for its counting numbers, as well as a recursive algorithm to compute the back-off rates more efficiently 
(Section \ref{sec:kmax}). Setting $k_{max}=2$ reduces the size $k_{max}$ clique approximation to the Bethe approximation of \cite{yun1}.

Third, we prove that the size $k_{max}$ clique approximation coincides with a Kikuchi approximation
(Section \ref{sec:kikuchi}).
As the Kikuchi approximation used to devise the IGBP algorithm of \cite{kai1} corresponds to
setting $k_{max}=n$, the size $n$ clique approximation gives a closed form expression for a
fixed point of the IGBP algorithm.

Fourth, an exact free energy approximation for chordal conflict graphs is introduced and is proven to
coincide with the size $n$ clique approximation (Section \ref{sec:chordal}). This implies that a fixed point of the
IGBP algorithm gives exact results on chordal conflict graphs.

Finally, simulation results are presented that compare the accuracy of the size $k_{max}$ clique approximation
with the LCS algorithm presented in \cite{vanhoudt_ton17} (Section \ref{sec:eval}). 
The main observation is that the LCS approximation is less accurate and less robust for denser conflict
graphs compared to the size $k_{max}$ clique approximation.

Before presenting the above results (in Sections \ref{sec:cliquebelief} to \ref{sec:eval}), we start with a model description in Section \ref{sec:model} and 
a basic introduction on (region-based) free energy approximations in Section \ref{sec:free}.
Conclusions are drawn in Section \ref{sec:conc}.

We end this section by noting that more advanced free energy approximations have very recently
and independently been proposed by other researchers to approximation the back-off rates in CSMA networks. 
More specifically, in \cite{swamy_SPCOM} the authors proposed the Kikuchi approximation induced by all
the maximal cliques of the conflict graph. In Section \ref{sec:kikuchi} of this paper we prove that this approximation
coincides with the size $n$ clique approximation. Further, the authors also prove the exactness
of the maximal clique based Kikuchi approximation on chordal conflict graphs as is done in 
Section \ref{sec:chordal} in this paper. In \cite{swamy_arXiv2017} the authors generalized
their work using the region-based free energy framework of \cite{yedidia1} and also consider an approximation
that includes the $4$-cycles as regions. 

\section{System model}\label{sec:model}

The ideal CSMA model considers a fixed set of $n$ links where a link is said to be active if a packet is being 
transmitted on the link and inactive otherwise. Whether two links $i$ and $j$ can be active simultaneously is determined by the undirected
conflict graph $G=(V,E)$, where $V$ is the set of $n$ links. If $(i,j) \in E$ then link $i$ and $j$ cannot be active at the
same time. The time that a link remains active is represented by an independent and identically distributed random variable
with mean one, after which the link becomes inactive. When link $i$ becomes inactive it starts a back-off period, the length
of which is an independent and identically distributed random variable with mean $1/\nu_i$.  When the back-off period
of link $i$ ends and none of the neighbors of $i$ in $G$ are active, link $i$ becomes active; otherwise a new back-off period
starts. Note, in such a setting sensing is assumed to be instantaneous and therefore no collisions occur. Also, links
attempt to become active all the time, which corresponds to considering a saturated network\footnote{It is worth noting that a considerable body of work exists that considers unsaturated CSMA networks where each link maintains its own buffer to store packets that arrive according to a Poisson
process (e.g., \cite{shah1,ni1,bouman1,laufer1,cecchi1}).}.
Hence, the ideal CSMA model is fully characterized by the conflict graph $G=(V,E)$ and the vector of back-off rates
$(\nu_1,\ldots,\nu_n)$. 

While the set of links, interference graph and target throughputs are all fixed in the ideal CSMA network,
the results presented in this paper are still meaningful in a network that undergoes gradual changes as in such case
the proposed approximation for the back-off rates can be recomputed at regular times in a {\it fully distributed} manner.

Let $x_i = 1$ if node $i$ is active and $0$ otherwise, for $i=1,\ldots,n$.
Define $f_i(x_i) = \nu_i^{x_i}$ and $f_{(i,j)}(x_i,x_j) = 1-x_ix_j$, for $i,j \in \{1,\ldots,n\}$. 
It is well known \cite{boorstyn1,vandeven3} that  the probability that the network 
is in state $(x_1,\ldots,x_n)$ at time $t$ converges to
\begin{align}\label{eq:pf}
p(x_1&,\ldots,x_n) = \nonumber \\
& \frac{1}{Z} \left(\prod_{i=1}^n f_i(x_i) \right) \left(\prod_{(i,j)\in E} f_{(i,j)}(x_i,x_j)\right),
\end{align}
for $(x_1,\ldots,x_n) \in \{0,1\}^n$ as $t$ tends to infinity, where $Z$ is the normalizing constant. 
Note the factors $f_{(i,j)}(x_i,x_j)$ make sure that
the probability of being in state $(x_1,\ldots,x_n)$ is zero whenever two neighbors in $G$ would be active at the same time.

The throughput of link $i$ is simply given by the marginal probability
\[p_i(1) = \sum_{x \in \{0,1\}^n} p(x) 1_{\{x_i = 1\}}.\]
The focus in this paper is not on computing $p_i(1)$ given the back-off rates $(\nu_1,\ldots,\nu_n)$,
but on the inverse problem: how to set/estimate $(\nu_1,\ldots,\nu_n)$ such that a given target throughput vector
$(\phi_1,\ldots,\phi_n) \in \Gamma$ is achieved, that is, such that $p_i(1)=\phi_i$ for all $i$.
In \cite{jiang1} the set of achievable throughput vectors was shown to equal 
\begin{align}\label{eq:Gamma} 
\Gamma =  \left\{ \sum_{\vec x \in \Omega} \xi(\vec x) \vec x \middle| \sum_{\vec x \in \Omega} \xi(\vec x) =1, \xi(\vec x) > 0 \mbox{ for } \vec x \in \Omega \right\},
\end{align}
where $\Omega = \{(x_1,\ldots,x_n) \in \{0,1\}^n | x_i x_j = 0 \mbox{ if } (i,j) \in E\}$.
In other words, a throughput vector $\vec \theta$ is achievable if and only if it belongs to the interior of the
convex hull of the set $\Omega$.

\section{Free energy and region-based approximations}\label{sec:free}

We start with a brief introduction on (region-based) free energy  approximations and describe these in the
context of factor graphs, we refer the reader to \cite{yedidia1} for a more detailed exposition.  

A factor graph \cite{kschischang1} is a bipartite graph that contains a set of variable and
factor nodes that represents the factorization of a function. 
The factor graph associated with \eqref{eq:pf} contains a variable node for each variable $x_i$, for $i=1,\ldots,n$,
and $n+|E|$ factor nodes $f_a$: one for each factor $f_i$ and $f_{(i,j)}$. 
Further, a variable node $x_i$ is connected to a factor node $f_a$ if and only if
$x_i$ is an argument of $f_a$. As illustrated in Figure \ref{fig:factorgraph} the factor graph of \eqref{eq:pf} can be obtained from the conflict graph $G$ by labeling  node $i$ as $x_i$,
replacing each edge $(i,j)$ by a factor node $f_{(i,j)}$, by connecting $f_{(i,j)}$ to $x_i$ and $x_j$ and by adding the factor nodes $f_i$,
where $f_i$ is connected to $x_i$.  

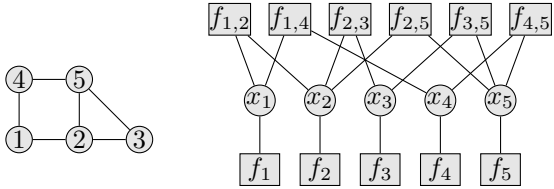
\begin{figure}[t]
\begin{center}
\begin{tikzpicture}[scale=0.8]
\tikzstyle{every node}=[circle, draw, fill=black!10,
                        inner sep=0pt, minimum width=10pt]

\node (1) at (1,0) []{$1$};
\node (2) at (2,0) []{$2$};
\node (3) at (3,0) []{$3$};
\node (5) at (2,1) []{$5$};
\node (4) at (1,1) []{$4$};

\draw [] (1) -- (2);
\draw [] (2) -- (3);
\draw [] (2) -- (5);
\draw [] (3) -- (5);
\draw [] (4) -- (5);
\draw [] (1) -- (4);

\tikzstyle{every node}=[circle, draw, fill=black!10,
                        inner sep=0pt, minimum width=10pt]

\node (f1) at (5,-0.5) [rectangle, minimum width=15pt, minimum height=12pt]{$f_1$};
\node (f2) at (6,-0.5) [rectangle, minimum width=15pt, minimum height=12pt]{$f_2$};
\node (f3) at (7,-0.5) [rectangle, minimum width=15pt, minimum height=12pt]{$f_3$};
\node (f4) at (8,-0.5) [rectangle, minimum width=15pt, minimum height=12pt]{$f_4$};
\node (f5) at (9,-0.5) [rectangle, minimum width=15pt, minimum height=12pt]{$f_5$};

\node (v1) at (5,0.65) []{$x_1$};
\node (v2) at (6,0.65) []{$x_2$};
\node (v3) at (7,0.65) []{$x_3$};
\node (v4) at (8,0.65) []{$x_4$};
\node (v5) at (9,0.65) []{$x_5$};

\node (f12) at (4.5,2) [rectangle, minimum width=15pt, minimum height=12pt]{$f_{1,2}$};
\node (f14) at (5.5,2) [rectangle, minimum width=15pt, minimum height=12pt]{$f_{1,4}$};
\node (f23) at (6.5,2) [rectangle, minimum width=15pt, minimum height=12pt]{$f_{2,3}$};
\node (f25) at (7.5,2) [rectangle, minimum width=15pt, minimum height=12pt]{$f_{2,5}$};
\node (f35) at (8.5,2) [rectangle, minimum width=15pt, minimum height=12pt]{$f_{3,5}$};
\node (f45) at (9.5,2) [rectangle, minimum width=15pt, minimum height=12pt]{$f_{4,5}$};

\draw [] (v1) -- (f1);
\draw [] (v2) -- (f2);
\draw [] (v3) -- (f3);
\draw [] (v4) -- (f4);
\draw [] (v5) -- (f5);

\draw [] (f12) -- (v1);
\draw [] (f12) -- (v2);
\draw [] (f14) -- (v1);
\draw [] (f14) -- (v4);
\draw [] (f23) -- (v3);
\draw [] (f23) -- (v2);
\draw [] (f25) -- (v5);
\draw [] (f25) -- (v2);
\draw [] (f35) -- (v3);
\draw [] (f35) -- (v5);
\draw [] (f45) -- (v4);
\draw [] (f45) -- (v5);

\end{tikzpicture}
\end{center}
\caption{Conflict graph $G$ and its associated bipartite factor graph}
\label{fig:factorgraph}
\end{figure}

Given a distribution $p$ on $\{0,1\}^n$ (as in \eqref{eq:pf}) with an associated factor graph with variable nodes $\{x_1,\ldots,x_n\}$
and factor nodes $\{f_1,\ldots,f_M\}$,
the Gibbs free energy $F(b)$, where $b$ is a distribution on $\{0,1\}^n$, is defined as $F(b) = U(b)-H(b)$, where
\[U(b) = - \sum_{x \in \{0,1\}^n} b(x) \sum_{a=1}^M \ln(f_a(x_a)), \]  
is the Gibbs average energy and the subset of the elements of $(x_1,\ldots,x_n)$ that are an argument of the function $f_a$ is 
denoted as $x_a$ and
\begin{align}\label{eq:GibbsH}
H(b) = - \sum_{x \in \{0,1\}^n} b(x) \ln(b(x)), 
\end{align} 
is the Gibbs entropy. For the factor graph of \eqref{eq:pf} we have
\begin{align*}
\sum_{a=1}^M \ln(f_a(x_a)) = \sum_{i=1}^n x_i \ln(\nu_i) +
\sum_{(i,j) \in E} \ln(1-x_i x_j). 
\end{align*}  
It is well-known that the Gibbs free energy associated with a factor graph is minimized when the distribution $b$ matches $p$.
Although the minimizer $p$ of the Gibbs free energy $F(b)$ may be known explicitly as in the CSMA setting, computing
marginal distributions of the form
\[p_S(x_S) = \sum_{x_i:i \not\in S} p(x_1,\ldots,x_n),\] 
where $S$ is a subset of $\{1,\ldots,n\}$,
is often computationally prohibitive. As such, approximations for the Gibbs free energy have been developed that
allow approximating marginal distributions of the form $p_S(x)$ at a (much) lower computational cost. 
Such approximations have also been used to attack the {\it inverse} problem which attempts to estimate
the model parameters (e.g., the back-off rates $\nu_i$ in CSMA or the couplings in the Ising model) 
given some values for some of the marginal distributions (e.g., the target throughput vector in CSMA or the
magnetizations and correlations in the Ising model). For the Bethe approximation this has led to 
explicit formulas for the approximate solution of the inverse problem for both the ideal CSMA model \cite{yun1} 
and the Ising model \cite{ricci1}. 

The class of free energy approximations that is used in this paper for the inverse problem is the class of 
region-based free energy approximations \cite{yedidia1}.
A region-based free energy approximation is characterized by a set $\mathcal{R}$ of regions and a 
counting number  $c_R$ for each $R \in \mathcal{R}$. 
Each region $R$ has an associated set of variables $\mathcal{V}_R$, which is a subset of the variable nodes in the factor graph,
and a set of factors denoted as $\mathcal{F}_R$, which is a subset of the factor nodes in the factor graph.
The following three conditions must be met for the sets $\mathcal{V}_R$ and $\mathcal{F}_R$. First, if $f_a \in \mathcal{F}_R$, then
the arguments of $f_a$ must belong to $\mathcal{V}_R$. Second, the set $\cup_{R \in \mathcal{R}} \mathcal{V}_R = \{x_1,\ldots,x_n\}$
and $\cup_{R \in \mathcal{R}} \mathcal{F}_R = \{f_1,\ldots,f_M\}$, in other words each variable node and factor node must
belong to at least one region. Third, 
the counting numbers $c_R$ are integers such that for each factor node $f_a$ and variable node $x_i$ we have 
\begin{align}\label{eq:valid}
\sum_{R \in \mathcal{R}} c_R 1_{\{f_a \in \mathcal{F}_R\}} = \sum_{R \in \mathcal{R}} c_R 1_{\{x_i \in \mathcal{V}_R\}} = 1.
\end{align}
For example for the Bethe approximation of \cite{yun1} one associates a single region $R$ with every node in the bipartite factor
graph. For the region $R$ associated with a factor node $f_a$ one sets $\mathcal{F}_R = \{f_a\}$, $\mathcal{V}_R = 
\{x_i | x_i \mbox{ is an argument of } f_a \}$ and $c_R = 1$. For the region $R$ corresponding 
to a variable node $x_i$ one sets $\mathcal{F}_R = \emptyset$,  $\mathcal{V}_R = \{x_i\}$ and
 $c_R = -d_i$, where $d_i$ is the number of neighbors of node $i$ in the conflict graph $G$
such that \eqref{eq:valid} holds. 

As in \cite{yedidia1} we denote sums of the form 
\[\sum_{x \in \{0,1\}^{|\mathcal{V}_R|}} g(x),\] 
where $g$ is a function from $\{0,1\}^{|\mathcal{V}_R|}$
to $\mathbb{R}$ and $R$ is a region, as $\sum_{x_R} g(x_R)$. 
Using this notation, the region-based free energy is a function of
 the set of beliefs $\{ b_R(x_R) | R \in \mathcal{R} \}$, where $b_R$ is a distribution on $\{0,1\}^{|\mathcal{V}_R|}$, and is
defined as 
\begin{align}\label{eq:FRbR}
F_R(\{b_R\}) = U_R(\{b_R\}) - H_R(\{b_R\}),
\end{align} 
where $U_R(\{b_R\})$
is the region-based average energy defined as
\begin{align}\label{eq:URbR}
U_\mathcal{R}(\{b_R\}) = - \sum_{R \in \mathcal{R}} c_R \sum_{x_R} b_{R}(x_R) \sum_{f_a \in \mathcal{F}_R} \ln(f_a(x_a)),
\end{align} 
and $H_R(\{b_R\})$ is the region-based entropy given by
\begin{align}\label{eq:HRbR}
H_\mathcal{R}(\{b_R\}) &= - \sum_{R \in \mathcal{R}} c_R \sum_{x_R} b_{R}(x_R) \ln(b_R(x_R)).
\end{align}
Note that the requirement that $f_a \in \mathcal{F}_R$ implies that
the arguments of $f_a$ must belong to $\mathcal{V}_R$ is necessary for \eqref{eq:URbR} to be well defined.

The beliefs $b_R(x_R)$ are used as approximations for the marginal probabilities 
$p_R(x_R) = \sum_{x_i \not\in \mathcal{V}_R} p(x_1,\ldots,x_n)$.
The approximation exists in finding the beliefs $b_R(x_R)$ such that the region-based free energy is minimized
over the set $\Delta_{\mathcal{R}}$ of {\it consistent} beliefs $b_R(x_R)$ defined as
\begin{align*}
\Delta_{\mathcal{R}} = & \left\{ \{b_R, R \in \mathcal{R} \} \middle| b_R(x_R) \geq 0, \sum_{x_R} b_R(x_R) = 1, \right. \\
&\hspace*{0.3cm} \left. \sum_{x_i \in \mathcal{V}_{R'}\setminus \mathcal{V}_R} b_{R'}(x_{R'}) = \sum_{x_j \in \mathcal{V}_R\setminus \mathcal{V}_{R'}} b_{R}(x_{R}) \right\}.
\end{align*}

We note that having a consistent set of beliefs does not imply that they are the marginals of a 
single distribution $b(x)$ on $\{0,1\}^n$ \cite[Section V.A]{yedidia1}.
Further, the average energy given by \eqref{eq:URbR} is known to be exact, that is, equal to the Gibbs free energy $U(p)$,
if $b_R(x_R) = p_R(x_R)$ for all $R$ and $x_R$. This condition is however not sufficient for 
 the region-based entropy to be exact (that is, equal to the Gibbs entropy $H(p)$) \cite{yedidia1}.

\section{Clique Belief}\label{sec:cliquebelief}
In this section we introduce the notion of clique believe and indicate how to select the
back-off rates to obtain a zero gradient point of the region-based
free energy under clique belief. These back-off rates, presented in Theorem \ref{th:backoff},
are used as an approximation for the vector of back-off rates that achieves
a given throughput vector $(\phi_1,\ldots,\phi_n)$. 

Clique belief is defined as the belief that all the nodes $i \in V$ with $x_i \in \mathcal{V}_R$ form a clique in the
conflict graph $G$ for any $R \in \mathcal{R}$, meaning the belief that any two nodes within a region $R$ are active at the same
time is zero.  More specifically, we define the set of clique beliefs $\Delta^{C}_{\mathcal{R}}$ 
as the set of beliefs $\{b_R\}$ for which  $b_R(x_R)$ has the form
\begin{align}\label{eq:cbelief}
b_R(x_R) = \left\{ \begin{array}{ll}
1-\sum_{i: x_i \in \mathcal{V}_R} \phi_i & \mbox{for } \sum_{x_i \in \mathcal{V}_R} x_i = 0,\\
\phi_i & \mbox{for }x_i = 1 \mbox{ and } \\
& \ \ \ \ \  \sum_{x_j \in \mathcal{V}_R} x_j = 1,\\
0 & \mbox{otherwise}. 
\end{array}\right.
\end{align}
for some set $\{\phi_1,\ldots,\phi_n\}$ with $\phi_i \geq 0$ and $\sum_{i:x_i \in \mathcal{V}_R} \phi_i < 1$ for all $R \in \mathcal{R}$.
Clique beliefs are clearly consistent, that is, $\Delta^{C}_{\mathcal{R}} \subseteq \Delta_{\mathcal{R}}$.

The next condition limits the set of region-based free energy approximations considered somewhat by
putting some minor conditions on the manner in which the regions are selected.

\begin{cond}\label{cond:region}
The set of regions $\mathcal{R}$ is such that $\mathcal{R} = \{R_{f_1},\ldots,R_{f_n}\} \cup \{R_{x_1},\ldots,R_{x_n}\} \cup \mathcal{R}'$ with
\begin{enumerate}
\item $\mathcal{V}_{R_{f_i}} = \mathcal{V}_{R_{x_i}} = \{x_i\}$, $\mathcal{F}_{R_{f_i}} = \{f_i\}$ and $\mathcal{F}_{R_{x_i}} = \emptyset$, 
for $i=1,\ldots,n$, 
\item for $R \in \mathcal{R}'$ we have $\emptyset \not= \mathcal{F}_R \subseteq \{f_{(i,j)} | (i,j) \in E\}$
and $\mathcal{V}_R = \{x_i | \exists j: f_{(i,j)} \in \mathcal{F}_R\}$,
\item  $c_{R_{x_i}} = 1-(1+\sum_{R \in \mathcal{R}'} c_R 1_{\{x_i \in \mathcal{V}_R\}})$ and $c_{R_{f_i}} = 1$, for $i = 1,\ldots,n$.
\end{enumerate}
\end{cond}

Note that this condition states that there are $2n$ special regions for which the set of variable nodes,
factor nodes as well as the counting numbers are fixed. For each of the remaining regions (that is, for each $R \in \mathcal{R}'$) 
the set of variable and factor nodes is determined by some nonempty set of edges, while its counting number can be chosen
arbitrarily as long as \eqref{eq:valid} holds.  

\begin{theorem}\label{th:backoff}
Let $\mathcal{R}$ be a set of regions that meets Condition \ref{cond:region}. For the zero gradient points
of the region-based free energy $F(\{b_R\})$ defined by \eqref{eq:FRbR} over the set $\Delta^{C}_{\mathcal{R}}$ of clique beliefs we have
\begin{align}\label{eq:backoff}
 \nu_i = \frac{\phi_i}{(1-\phi_i)^{1+c_{R_{x_i}}}} \prod_{\substack{R \in \mathcal{R}':\\ x_i \in \mathcal{V}_R}} \left( 1-\sum_{j:x_j \in \mathcal{V}_R} \phi_j \right)^{-c_R}.
\end{align}
\end{theorem}
\begin{proof}
Under clique belief the entropy given by \eqref{eq:HRbR} equals
\begin{align*}
H_\mathcal{R}&(\{b_R\}) = - \sum_{R \in \mathcal{R}} c_R \sum_{i: x_i\in \mathcal{V}_R} \phi_i \ln(\phi_i) \nonumber \\
&- \sum_{R \in \mathcal{R}} c_R \left( 1-\sum_{i:x_i\in \mathcal{V}_R} \phi_i\right) \ln(1-\sum_{i:x_i \in \mathcal{V}_R} \phi_i) \nonumber \\
\end{align*}
due to \eqref{eq:cbelief} and \eqref{eq:valid} yields 
\begin{align}\label{eq:HRbR2}
H_\mathcal{R}&(\{b_R\}) =  - \sum_{i=1}^n \phi_i \ln(\phi_i) \nonumber\\
&- \sum_{R \in \mathcal{R}} c_R \left( 1-\sum_{i:x_i \in \mathcal{V}_R} \phi_i\right) \ln(1-\sum_{i:x_i \in \mathcal{V}_R} \phi_i).
\end{align}
To determine $U_\mathcal{R}(\{b_R\})$ first note that $\ln(f_{(i,j)}(x_i,x_j))$ is zero unless $x_i = x_j = 1$. 
When  $x_i = x_j = 1$, we have $b(x_R) = 0$ if $f_{(i,j)} \in \mathcal{F}_R$ for $R \in \mathcal{R}'$. 
Using the common convention that $0 \ln(0) = \lim_{x \downarrow 0} x \ln(x) = 0$ yields
\[- \sum_{R \in \mathcal{R'}} c_R \sum_{x_R} b_{R}(x_R) \sum_{f_a \in \mathcal{F}_R} \ln(f_a(x_a)) = 0\]
As $\ln(f_i(0))=0$, $\ln(f_i(1))=\ln(\nu_i)$ and $b_{R_{f_i}}(1) = \phi_i$, one therefore finds
\begin{align}\label{eq:URbR2}
U_\mathcal{R}(\{b_R\}) = - \sum_{i \in V} \phi_i \ln(\nu_i).
\end{align}

Demanding that the partial derivatives $dF_\mathcal{R}(\{b_R\})/d\phi_i$ are equal to zero is equivalent to the requirement that
\begin{align*}
\ln(\nu_i) = &-\sum_{R \in \mathcal{R}:x_i \in \mathcal{V}_R} c_R \left( 1+\ln(1-\sum_{j:x_j\in \mathcal{V}_R} \phi_j) \right)\nonumber\\
& + 1+\ln(\phi_i). 
\end{align*}
The expression in \eqref{eq:backoff} therefore follows from \eqref{eq:valid}.
\end{proof}

Formula \eqref{eq:backoff} proposes an approximation for the back-off rates $\nu_i$ if the target throughputs of the links are
given by the vector $(\phi_1,\ldots,\phi_n)$ provided that $\sum_{i:x_i \in \mathcal{V}_R} \phi_i < 1$ for all $R \in \mathcal{R}$. 
Depending on the choice of the regions in $\mathcal{R}'$, this condition may be more restrictive than demanding that $(\phi_1,\ldots,\phi_n) \in \Gamma$. 
However for the size $k_{max}$ clique approximation introduced in the next section, the requirement $\sum_{i:x_i \in \mathcal{V}_R} \phi_i < 1$ 
holds for any $(\phi_1,\ldots,\phi_n) \in \Gamma$ as for the size $k_{max}$ clique approximation
each region $R$ corresponds to a clique and the sum of the throughputs of all the nodes belonging
to a clique is clearly bounded by one in $\Gamma$.

It is important to stress that formula \eqref{eq:backoff} often leads to a distributed computation of $\nu_i$ as node $i$ only needs to know
its own target throughput, the target throughput of any node $j$ sharing a region with $i$ (that is, any $j$
for which there exists an $R \in \mathcal{R}'$ such that $x_i,x_j \in \mathcal{V}_R$) as well as the counting numbers
$c_R$ for the regions $R$ to which it belongs. The size $k_{max}$ clique  approximation presented in the next sections is such that 
 two nodes only belong to the same region if they are neighbors in the conflict graph $G$ and the required counting numbers
can be computed from the subgraph induced by a node and its one hop neighborhood.

Thus, for the size $k_{max}$ clique approximation  a node can compute its approximate back-off rate  using information from its {\it one-hop neighbors only for any feasible throughput vector}.

\section{Size $k_{max}$ clique approximation}\label{sec:kmax}
In this section we introduce a region-based free energy approximation for general conflict graphs $G$, called
the size $k_{max}$ clique approximation and present explicit expressions for the counting numbers.
We start by considering two special cases.

\subsection{Bethe approximation}
A first special case is to define $\mathcal{R}$ such that 
Condition \ref{cond:region} is met and setting $\mathcal{R}' = \{R_{(i,j)} | (i,j) \in E\}$ such that
$V_{R_{(i,j)}} =\{x_i,x_j\}$ and $F_{R_{(i,j)}} = \{f_{(i,j)}\}$.
The associated counting numbers are $c_{R_{(i,j)}} = 1$, which implies that $c_{R_{x_i}} = -d_i$,  where $d_i$ denotes the number of
neighbors of $i$ in $G$.

The entropy given in \eqref{eq:HRbR2} therefore becomes
\begin{align*}
H_\mathcal{R}&(\{b_R\}) = - \sum_{(i,j) \in E} \left( 1 -\phi_i-\phi_j \right) \ln(1-\phi_i-\phi_j) \\ 
&+ \sum_{i=1}^n \left[(d_i-1)(1-\phi_i)\ln(1-\phi_i) - \phi_i \ln(\phi_i) \right]
\end{align*} 
and \eqref{eq:backoff} implies that the back-off rate, denoted as $\nu_i^{(2)}$, should be set as
\begin{align}\label{eq:backoffBethe}
 \nu_i^{(2)} =  \frac{\phi_i(1-\phi_i)^{d_i-1}}{\prod_{(i,j)\in E} (1-\phi_i-\phi_j)}.
\end{align}
The above expression corresponds to the Bethe approximation for CSMA networks proposed in \cite{yun1}. 
It is worth noting that \eqref{eq:backoffBethe} is 
a fixed point of the inverse belief propagation (IBP) algorithm presented in \cite{kai1}.
More specifically, the update rule in \cite[Section IV.C]{kai1} can be written as
\[\frac{m_{ji}(0)}{m_{ji}(1)} = 1 + \frac{m_{ij}(0)}{m_{ij}(1)} \frac{\phi_j}{1-\phi_j},\]
where $\phi_j$ is the target
throughput of node $j$. It is easy to check that this update rule has $m_{ij}(0)/m_{ij}(1)=(1-\phi_j)/(1-\phi_i-\phi_j)$
as a fixed point and if we plug this into Equation (8) of \cite[Section IV.C]{kai1} we obtain
\eqref{eq:backoffBethe}.   

\subsection{Triangle approximation}

A second special case, called the {\it triangle} approximation, is obtained by extending the set of regions $\mathcal{R}'$
as defined in the previous subsection with the regions $\{R_{(i,j,k)} | (i,j),(i,k),(j,k) \in E\}$ 
such that $V_{R_{(i,j,k)}} =\{x_i,x_j,x_k\}$ and $F_{R_{(i,j)}} = \{f_{(i,j)},f_{(i,k)},f_{(j,k)}\}$.
The counting numbers are now set as $c_{R_{(i,j,k)}} = 1$ and $c_{R_{(i,j)}}=1-t_{i,j}$ such that $c_{R_{x_i}}=t_i-d_i$, 
where $t_i$ and $t_{i,j}$ are the number of triangles in $G$ that contain
node $i$ and edge $(i,j)$, respectively. Note that these counting numbers obey the requirement given
in \eqref{eq:valid} as $\sum_{j \in \mathcal{N}_i} t_{i,j} = 2 t_i$.

For the triangle approximation the back-off rates, denoted as $\nu_i^{(3)}$, given by  \eqref{eq:backoff} correspond to
\begin{align}
 \nu_i^{(3)} =  \frac{\phi_i (1-\phi_i)^{d_i-1}
\prod_{(i,j)\in E} (1-\phi_i-\phi_j)^{t_{i,j}-1}}{(1-\phi_i)^{t_i}\prod_{(i,j,k)\in \Delta_E} (1-\phi_i-\phi_j-\phi_k)},
\end{align}
where $\Delta_E$ denotes the set of triangles in $E$.

\subsection{General case}
The idea behind the Bethe and triangle approximation can be generalized to cliques of larger sizes, at the
expense of an increased complexity to compute the back-off rates. 
In this section the set $\mathcal{R}'$ corresponds to the set of all the cliques $K$ in the conflict graph $G$ with
a size in $\{2,\ldots,k_{max}\}$, where $k_{max}$ is a predefined maximum allowed clique size. Note that setting $k_{max}=2$ and $3$ corresponds to the previous two approximations.
If $K$ is a clique of size $k \in \{2,3,\ldots,k_{max}\}$ and $R(K)$ its associated region, then
$V_{R(K)} = \{x_i | i \in K\}$ and $F_{R(K)} = \{f_{(i,j)} | i,j \in K\}$. 

The counting number $c_{R(K)}=1$ for any region $R(K)$ associated with a clique $K$ of size $k_{max}$.
For a region $R(K)$ corresponding to a clique $K$ of size $k$ with $1 \leq k < k_{max}$, we set $c_{R(K)}$ 
\[ c_{R(K)} = 1_{\{k > 1\}} - \sum_{ R' \in \mathcal{R}'} c_{R'} 1_{\{V_{R(K)} \subset \mathcal{V}_{R'}\} }.\]
Note, for any maximal clique $K$ of size $2 \leq k \leq k_{max}$, we have $c_{R(K)} = 1$ irrespective of its size.  
The next proposition provides an explicit expression for $c_{R(K)}$.

\begin{theorem}
If $K$ is a clique of size $k \in \{1,\ldots,k_{max}\}$ then
\begin{align}\label{prop:cRK}
c_{R(K)} = 1_{\{k > 1\}} + \sum_{s=k+1}^{k_{max}} (-1)^{s-k} n_{K,s},
\end{align}
where $n_{K,s}$ denotes the number of cliques $K'$ in $G$ with $|K'| = s$ and $K \subset K'$.
\end{theorem}
\begin{proof}
The result clearly holds for $k=k_{max}$. 
We use backward induction on $|K|=k$ to find
\begin{align}
-&c_{R(K)}+1_{\{k > 1\}} = \nonumber\\
& \sum_{u=k+1}^{k_{max}} \sum_{\substack{cliques \ K':\\K \subset K', |K'|=u}} \left( 1  
+\sum_{s=u+1}^{k_{max}} (-1)^{s-u} n_{K',s}\right)=\nonumber \\
& \sum_{u=k+1}^{k_{max}} n_{K,u} + \sum_{u=k+1}^{k_{max}} \sum_{s=u+1}^{k_{max}} (-1)^{s-u} \hspace*{-0.6cm}
\sum_{\substack{cliques \ K':\\K \subset K', |K'|=u}}  n_{K',s} \label{eq:cRKstep1}
\end{align}
The latter sum equals $\binom{s-k}{u-k} n_{K,s}$ as we can pick $u-k$ elements from the $s-k$ elements not belonging to $K$ of a
size $s$ clique that contains $K$ to obtain a $K'$. Hence, switching sums in \eqref{eq:cRKstep1} yields
\begin{align}
-&c_{R(K)} + 1_{\{k > 1\}} =\nonumber\\
&\sum_{s=k+1}^{k_{max}} n_{K,s} + \sum_{s=k+2}^{k_{max}} \sum_{u=k+1}^{s-1} (-1)^{s-u} \binom{s-k}{u-k} n_{K,s}= \nonumber\\
& \sum_{s=k+1}^{k_{max}} \left( 1 + (-1)^{s-k}\sum_{z=1}^{s-k-1} (-1)^{-z} \binom{s-k}{z}\right)n_{K,s}\label{eq:cRKstep2}
\end{align}
By noting that $\sum_{i=0}^n (-1)^i \binom{n}{i} = 0$, \eqref{eq:cRKstep2} becomes 
\begin{align*}
&c_{R(K)} = 1_{\{k > 1\}} +\nonumber \\
&\ - \sum_{s=k+1}^{k_{max}}  \left( 1 + (-1)^{s-k}(-1-(-1)^{s-k})\right)n_{K,s}\nonumber\\
&\ = 1_{\{k> 1\}} + \sum_{s=k+1}^{k_{max}}   (-1)^{s-k}n_{K,s},
\end{align*}
as required.
\end{proof}

Note that increasing $k_{max}$ by one simply adds one additional term to $c_{R(K)}$ in \eqref{prop:cRK}, 
which allows us  to compute the back-off rates of the size $k_{max}$ clique approximation in a recursive
manner as follows.

\begin{corollary}
Let $\nu_i^{(k_{max})}$ be the back-off rate for node $i$ corresponding with the size $k_{max}$ clique approximation, then
\begin{align}\label{eq:recurs}
 &\nu_i^{(k_{max})}  = \nu_i^{(k_{max}-1)} \cdot \nonumber \\
 &\prod_{\substack{cliques \ K \ in \ G:\\ i \in K, 1 \leq |K| \leq k_{max}}} 
\left(1-\sum_{s\in K} \phi_s \right)^{-n_{K,k_{max}}(-1)^{k_{max}-|K|}},\end{align}
where $\nu_i^{(1)}= \phi_i/(1-\phi_i)$ and $n_{K,k_{max}}$ denotes the number of size $k_{max}$ cliques $K'$ in $G$ with $K \subset K'$.
\end{corollary}
\begin{proof}
The result follows from \eqref{eq:backoff} when combined with \eqref{prop:cRK}.
\end{proof}

We now briefly discuss the complexity to compute the back-off rate of node $i$ when using the $k_{max}$ clique approximation.
Node $i$ can be part of at most $\min(2^{d_i},d_i^{k_{max}-1})$ cliques $K$ with $|K| \leq k_{max}$, where $d_i$
is the number of neighbors of $i$ in $G$. This set can be computed by first listing the $d_i$ size $2$ cliques
containing $i$. Having obtained the set of size $k$ cliques that contain $i$, the set of size $k+1$ cliques is found by
considering all its one element extensions. By using an ordered list $\{i_1,\ldots,i_{k-1}\}$ of the $k-1$ other nodes 
belonging to a size $k$ clique containing $i$, only one element extensions with a node $j > i_{k-1}$ need to be considered
and the creation of identical cliques of size $k+1$ is avoided. Having obtained the list of cliques $K$ that contain $i$
with $|K| \leq k_{max}$, the back-off rate given by \eqref{eq:recurs} can be readily computed by noting that 
\begin{align*}
 &\nu_i^{(k_{max})}  = \nu_i^{(k_{max}-1)} \prod_{\substack{cliques \ K' \ in \ G:\\ i \in K', |K'| = k_{max}}}  \cdot \nonumber \\
 & \prod_{K \subseteq K': i \in K} 
\left(1-\sum_{s\in K} \phi_s \right)^{(-1)^{k_{max}-|K|+1}}.
\end{align*}

\section{Kikuchi approximations}\label{sec:kikuchi}
The IGBP algorithm of \cite{kai1} is a message passing algorithm to estimate the
back-off rates to achieve a given throughput vector $(\phi_1,\ldots,\phi_n) \in \Gamma$.
This algorithm is based on a so-called Kikuchi free energy approximation. In this section 
we show that the size $k_{max}$ clique approximation also coincides with a Kikuchi approximation.
In fact for $k_{max}=n$ this Kikuchi approximation corresponds to the one associated to the IGBP
algorithm.  As such the expression for the back-off rates of the size $n$ clique
approximation gives us an explicit expression for a fixed point of the IGBP algorithm 
in \cite[Section VI.B]{kai1} due to \cite[Section VII]{yedidia1} (as the Bethe approximation did for the IBP algorithm). 

In a Kikuchi approximation (see \cite[Appendix B]{yedidia1} for more details) 
the set of regions $\tilde{\mathcal{R}}$ can be written as $\tilde{\mathcal{R}} = \bigcup_{i=0}^s \mathcal{R}_i$, for some $s$. 
We state that a region $R$ is a subset of a region $R'$ if $\mathcal{V}_R \subseteq \mathcal{V}_{R'}$ and $\mathcal{F}_R \subseteq F_{R'}$.
The regions in $\mathcal{R}_0$ fully characterize a Kikuchi approximation as follows. 
The regions in $\mathcal{R}_{i+1}$, for $i=0,\ldots,s$, are constructed from the sets $\mathcal{R}_0,\ldots,\mathcal{R}_i$ by taking all
the different intersections $R_i \cap R_j \not= \emptyset$, with $R_i \not\subseteq R_j$ and $R_j \not\subseteq R_i$,
of the regions $R_i \in \mathcal{R}_i$ with the regions $R_j \in \bigcup_{k=0}^i \mathcal{R}_k$ and subsequently removing the
sets $R \in \mathcal{R}_{i+1}$ for which there exists an $R' \in \mathcal{R}_{i+1}$ with $R \subseteq R'$. 
Note, $s$ is the smallest integer such that $\mathcal{R}_{s+1}$ is empty.
The counting number $\tilde c_R$ of region $R \in \mathcal{R}_i$ in a Kikuchi approximation is given by 
\[\tilde c_R = 1-\sum_{R' \in \tilde{\mathcal{R}}: R \subset R'} \tilde c_{R'} = 1-\sum_{R' \in \cup_{k=0}^{i-1}
\mathcal{R}_k: R \subset R'} \tilde c_{R'},\] 
as a region $R \in \mathcal{R}_i$ cannot be a subset of a region  $R' \in \mathcal{R}_j$ with $j \geq i$
(since this would imply the existence of a superset of $R$ in $\mathcal{R}_i$). 

\begin{theorem}\label{th:Kikuchi}
The size $k_{max}$ clique approximation coincides with a Kikuchi approximation with $\mathcal{R}_0 = \{R_{f_1},\ldots,R_{f_n}\} 
\cup \{R(K) | K \in \mathcal{K}_G(k_{max})\}$, where $\mathcal{K}_G(k_{max})$ is the union of the set of all the cliques of size $k_{max}$
and the set of the maximal cliques of size $k \in \{2,\ldots, k_{max}-1\}$ in $G$. 
\end{theorem}
\begin{proof}
See Appendix \ref{app:Kikuchi_proof}.
\end{proof}
  
The above theorem shows that the maximal clique based Kikuchi approximation considered in \cite{swamy_SPCOM, swamy_arXiv2017} coincides 
with the size $n$ clique approximation. We note that no explicit expression for the counting numbers or
a recursive scheme similar to \eqref{eq:recurs} is presented in \cite{swamy_SPCOM, swamy_arXiv2017}.

\section{Chordal conflict graphs}\label{sec:chordal}
In this section we establish two results: (a) we show that the exact explicit expressions for the
back-off rates for chordal conflict graphs, presented in \cite{vanhoudt_ton17}, corresponds to 
a zero gradient point of a region-based free energy approximation defined for chordal conflict graphs only and 
(b) we prove that the size $n$ clique approximation coincides with this chordal free energy approximation.
This implies that the size $n$ clique approximation (and therefore also a fixed point of the
IGBP algorithm of \cite{kai1}) provide exact results for chordal conflict graphs $G$.

A graph $G$ is chordal if and only if all cycles consisting of more than $3$ nodes have a {\it chord}. A chord of a cycle 
is an edge joining two nonconsecutive nodes of the cycle. Let $\mathcal{K}_G = \{K_1,\ldots,K_m\}$ be the set of maximal cliques of $G$. 
A clique tree $T=(\mathcal{K}_G,\mathcal{E})$ is a tree in which the nodes correspond
to the maximal cliques and the edges are such that the subgraph of $T$ induced by the maximal cliques that contain the node $v$
is a subtree of $T$ for any $v \in V$. A graph $G$ is chordal if and only if it has at least one clique tree 
(see Theorem 3.1 in \cite{blair1}).

For chordal conflict graphs $G$ we can define a region-based free energy approximation, called the
{\it chordal} region-based free energy approximation, by making use of {\it any}
clique tree $T = (\mathcal{K}_G,\mathcal{E})$ of $G$ in the following manner.
We define a set $\mathcal{R}$ containing 
$2n+2|\mathcal{K}_G|-1$ regions: one region $R_K$ for each maximal clique $K \in \mathcal{K}_G$, one region $R_{(K,K')}$ for each edge $(K,K') \in 
\mathcal{E}$, one region $R_{f_i}$ for each factor node $f_i$ and one region $R_{x_i}$ for each variable node $x_i$. 
Let $\mathcal{V}_R$ and $\mathcal{F}_R$ denote the set of variable and factor nodes associated with region $R \in \mathcal{R}$, then
$V_{R_K} = \{x_i | i \in K\}$ and 
\begin{align*}
F_{R_K} &= \{f_{(i,j)}| i,j \in K\},
\end{align*}
for $K \in \mathcal{K}_G$, $V_{R_{(K,K')}} =\{x_i | i \in K \cap K'\}$ and
\begin{align*}
F_{R_{(K,K')}} &= \{f_{(i,j)}| i,j \in K \cap K'\},
\end{align*}
for $(K,K') \in \mathcal{E}$ and $V_{R_{x_i}}  = V_{R_{f_i}} =\{x_i\}$,  $F_{R_{x_i}} =\emptyset$ and $F_{R_{f_i}} =\{f_i\}$.
The counting numbers $c_R$ are defined as follows: $c_{R_K} = c_{R_{f_i}} =1$ and $c_{R_{(K,K')}} = c_{R_{x_i}} = -1$.

Note the set of regions $\mathcal{R}$ fulfills Condition \ref{cond:region}, therefore under clique belief we have
$U_\mathcal{R}(\{b_R\}) = - \sum_{i \in V} \phi_i \ln(\nu_i)$.
As the nodes in $K$ and $K\cap K'$ form a clique, the clique belief matches the exact marginal probabilities $p_R(x_R)$
for each $R \in \mathcal{R}$ when $\phi_i = \sum_{x \in \{0,1\}^n} p(x) 1_{\{x_i = 1\}}$ for $i \in V$. 

As the believes $b_{R_{f_i}}(x_i)$ and $b_{R_{x_i}}(x_i)$ are the same and $c_{R_{f_i}} = -c_{R_{x_i}}$ these
regions cancel each other in the expression for the entropy. Thus for the entropy we have
\begin{align}\label{eq:HRbRchor}
H_\mathcal{R}&(\{b_R\}) = - \sum_{i \in V} \phi_i \ln(\phi_i) \nonumber \\
&+ \sum_{(K,K') \in \mathcal{E}} \left(1-\sum_{s \in K\cap K'} \phi_s \right) \ln(1-\sum_{s \in K\cap K'} \phi_s ) \nonumber \\ 
&- \sum_{K \in \mathcal{K}_G} \left(1-\sum_{s \in K} \phi_s \right) \ln(1-\sum_{s \in K} \phi_s ).
\end{align}
As noted before, even when the believes are equal to the exact marginal probabilities, 
the region-based entropy is in general not exact. Below we prove that the entropy (and therefore
also the energy) is exact in this particular case by leveraging existing results on the junction 
graph approximation method \cite[Appendix A]{yedidia1}.

\begin{theorem}\label{th:chordal}
The expression for the region-based entropy $H_\mathcal{R}(\{b_R\})$ given by \eqref{eq:HRbRchor} 
is equal to the Gibbs entropy $H(p)$ defined by \eqref{eq:GibbsH}. Further,  $H(p) = - \ln(\frac{1}{Z}\prod_{i\in V} \nu_i^{p_i(1)})$ and
if $x \in \{0,1\}^n$ such that $x_ix_j = 0$ if $(i,j) \in E$ ($p(x)=0$ otherwise), we have 
\begin{align}\label{eq:px}
 p&(x) = \frac{\prod_{K \in \mathcal{K}_G} (\vartheta_1(K)+\vartheta_0(K))}{\prod_{(K,K') \in \mathcal{E}}
(\vartheta_1(K\cap K')+\vartheta_0(K\cap K'))}, 
\end{align}
where 
$\vartheta_1(S) = \sum_{i\in S} p_i(1) 1_{\{x_i=1, \sum_{i \in S} x_i=1\}}$
and
$\vartheta_0(S) = \left(1-\sum_{s\in S} p_s(1)\right) 1_{\{\sum_{i \in S} x_i=0\}}$.
\end{theorem}
\begin{proof}
See Appendix \ref{app:chordal_proof}.
\end{proof}

\begin{corollary}
For a chordal conflict graph $G$ with clique tree $(\mathcal{K}_G,\mathcal{E})$, the normalizing constant $Z$ is given by
\begin{align}\label{eq:Z}
 Z = \frac{\prod_{(K,K') \in \mathcal{E}}  \left(1-\sum_{s\in K\cap K'} p_s(1)\right)}{\prod_{K \in \mathcal{K}_G} \left(1-\sum_{s\in K} p_s(1)\right)}, 
\end{align}
and the back-off rate $\nu_i$ obeys
\begin{align}\label{eq:nui}
 \nu_i = p_i(1) \frac{\prod_{(K,K') \in \mathcal{E}: i \in K \cap K'}  \left(1-\sum_{s\in K\cap K'} p_s(1)\right)}{\prod_{K \in \mathcal{K}_G: i \in K} 
\left(1-\sum_{s\in K} p_s(1)\right)}, 
\end{align}
where the marginal probability $p_i(1)$ is the throughput of link $i$.
\end{corollary}
\begin{proof}
The expression for $Z$ follows from equating \eqref{eq:pf} and \eqref{eq:px} with $x=(0,\ldots,0)$. 
The back-off rate $\nu_i$ is found in the same way using $x$ with $x_i=1$ and $x_j = 0$ for $i \not= j$.  
\end{proof}

Note the above formula for the back-off rate $\nu_i$ for chordal conflict graphs $G$ was derived earlier 
in \cite{vanhoudt_ton17} and corresponds to \eqref{eq:backoff}.

\begin{theorem}\label{th:coin}
When the conflict graph $G$ is chordal the Kikuchi approximation with $\mathcal{R}_0 = \{R_{f_1},\ldots,R_{f_n}\} 
\cup \{R(K) |K \in \mathcal{K}_G\}$, where $\mathcal{K}_G$ is the set of the maximal cliques in $G$,
coincides with the chordal region-based free energy approximation (defined for chordal conflict graphs only). 
\end{theorem}
\begin{proof}
See Appendix \ref{app:coin}.
\end{proof}

\begin{corollary}
When $G$ is chordal, the back-off rates given by \eqref{eq:backoff} for the
size $k_{max}=n$ clique approximation or the Kikuchi approximation 
with $\mathcal{R}_0 = \{R_{f_1},\ldots,R_{f_n}\} \cup \{R(K) |K \in \mathcal{K}_G\}$, 
where $\mathcal{K}_G$ is the set of the maximal cliques in $G$, are both equal to \eqref{eq:nui}
and are therefore exact.
\end{corollary}

The exactness of the above Kikuchi approximation for chordal conflict graphs was established 
independently in \cite{swamy_SPCOM, swamy_arXiv2017}.

\section{Experimental evaluation}\label{sec:eval}
To study the accuracy of the size $k_{max}$ clique approximation we perform simulation
experiments similar to one ones presented in \cite[Section IV.E]{kai1} for IGBP, except that we
consider a different set of conflict graphs and also compare with the local
chordal subgraph (LCS) approximation presented in \cite{vanhoudt_ton17}. More
specifically, we simulate the ideal CSMA model with the back-off rates set as estimated by
each approximation method and compute the mean relative error between the
given target throughputs (used as input by the approximation method) and 
the throughputs observed during simulation. 

The set of conflict graphs considered is similar to the ones used in \cite{vanhoudt_ton17}:
the $n$ nodes of the conflict graph are placed randomly in a square of size $1$ and there
exists an edge between two nodes if and only if the Euclidean distance between them is less 
than some threshold $R$. In the experiments we set $n=100$ nodes with $R$ values equal to
$0.15, 0.2$ and $0.25$. At this point we also note that node $i \in V$ requires the same input to compute its back-off rate
irrespective of whether it uses the size $k_{max}$ clique approximation or the LCS approximation:
it needs to construct the subgraph $G_i=(V_i,E_i)$ induced by node $i$ and its $d_i$ neighbors.
Hence both approximations can be implemented in a fully distributed manner. 

\begin{figure}[t]
\center
\includegraphics[width=0.4\textwidth]{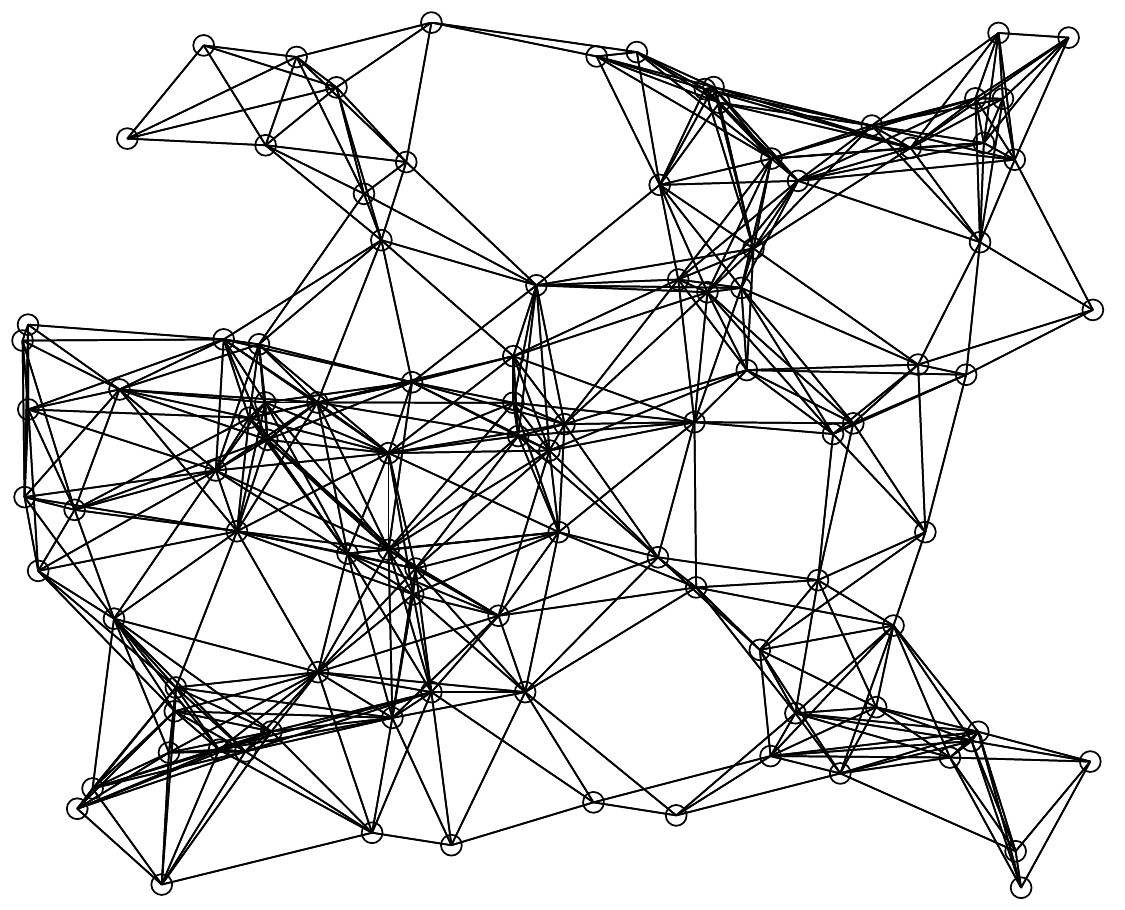}
\caption{Conflict graph with $n=100$ used in Figure \ref{fig:relerror} for $R = 0.20$.}
\label{fig:R020}
\end{figure}

In a first set of experiments the target throughput of each node was set equal to $\phi$ divided
by the size of the largest clique in the graph, where $\phi$ equals $0.55, 0.7$ and $0.85$.
Note that for $\phi > 1$ this vector does not belong to $\Gamma$, the set of achievable throughput vectors. 
Figure \ref{fig:relerror} presents the mean relative error of the LCS and size $k_{max}$ clique approximation
for various combinations of $R$ and $\phi$ (the corresponding conflict graph for $R=0.20$ is shown in Figure \ref{fig:R020}). 
For each value of $R$, the largest value for $k_{max}$ presented in this
figure corresponds to the size of the largest clique in the graph, that is, the same result is obtained when
setting $k_{max}=n$. Recall that we showed earlier that 
the size $n$ clique approximation coincides with a fixed point of the IGBP algorithm of \cite{kai1}.
Thus, the rightmost bar in Figure \ref{fig:relerror} corresponds to a fixed point of the IGBP algorithm.
Figure \ref{fig:relerror}  shows that the size $k_{max}$ clique approximation is more accurate that the LCS approximation in this setting, 
except for small $k_{max}$, at the expense of being more complex.

\begin{figure}[t]
\center
\includegraphics[width=0.45\textwidth]{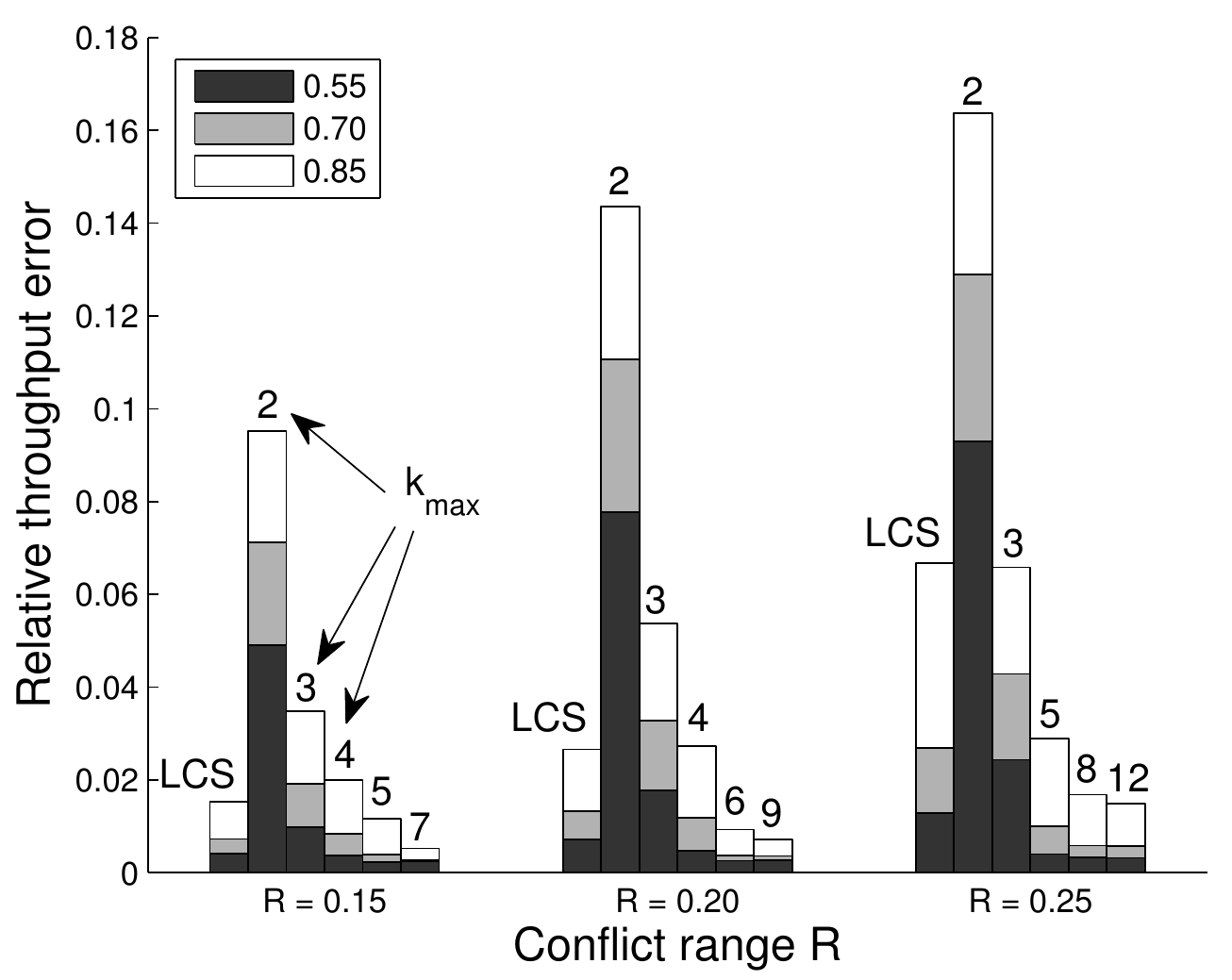}
\caption{Mean relative error in the achieved throughput of the LCS and size $k_{max}$ clique approximation
in a conflict graph with $n=100$ nodes for $\phi = 0.55, 0.7$ and $0.85$ and $R = 0.15, 0.2$ and $0.25$.}
\label{fig:relerror}
\end{figure}

We further note that the relative errors grow as the graph becomes more dense (increasing $R$) and this growth seems
more pronounced for the LCS approximation. We also note that the approximation becomes
worse as the target throughput of the links increases (increasing $\phi$). Nevertheless the mean relative
error of the size $n$ clique approximation remains below $2\%$ in all cases. This is somewhat higher
than the values reported in \cite{kai1} for IGBP, but this is mostly due to the fact that more
dense conflict graphs are considered here (for $R=0.15$ we have $301$ edges, while for $R=0.25$ 
we have as many as $788$ edges). 

While the results in Figure \ref{fig:relerror} are based on three conflicts graphs only, 
Figure \ref{fig:relerror2} compares the accuracy of the LCS and size $k_{max}$ approximations 
with $k_{max} = 2, 5, n$ on a set of $75$ conflict graphs: $25$ for each of the three $R$ values.
The target throughput of node $i$ is set equal to $0.85/(1+d_i)$, meaning not all nodes
have the same target throughput (as opposed to Figure \ref{fig:relerror}).
For each $R$ value and approximation method considered, Figure \ref{fig:relerror2} depicts the
mean relative throughput error (obtained by simulation) for the $2$ conflict graphs 
that  resulted in the smallest and largest mean relative throughput error, as well as the average
taken over the $25$ conflict graphs.

\begin{figure}[t]
\center
\includegraphics[width=0.45\textwidth]{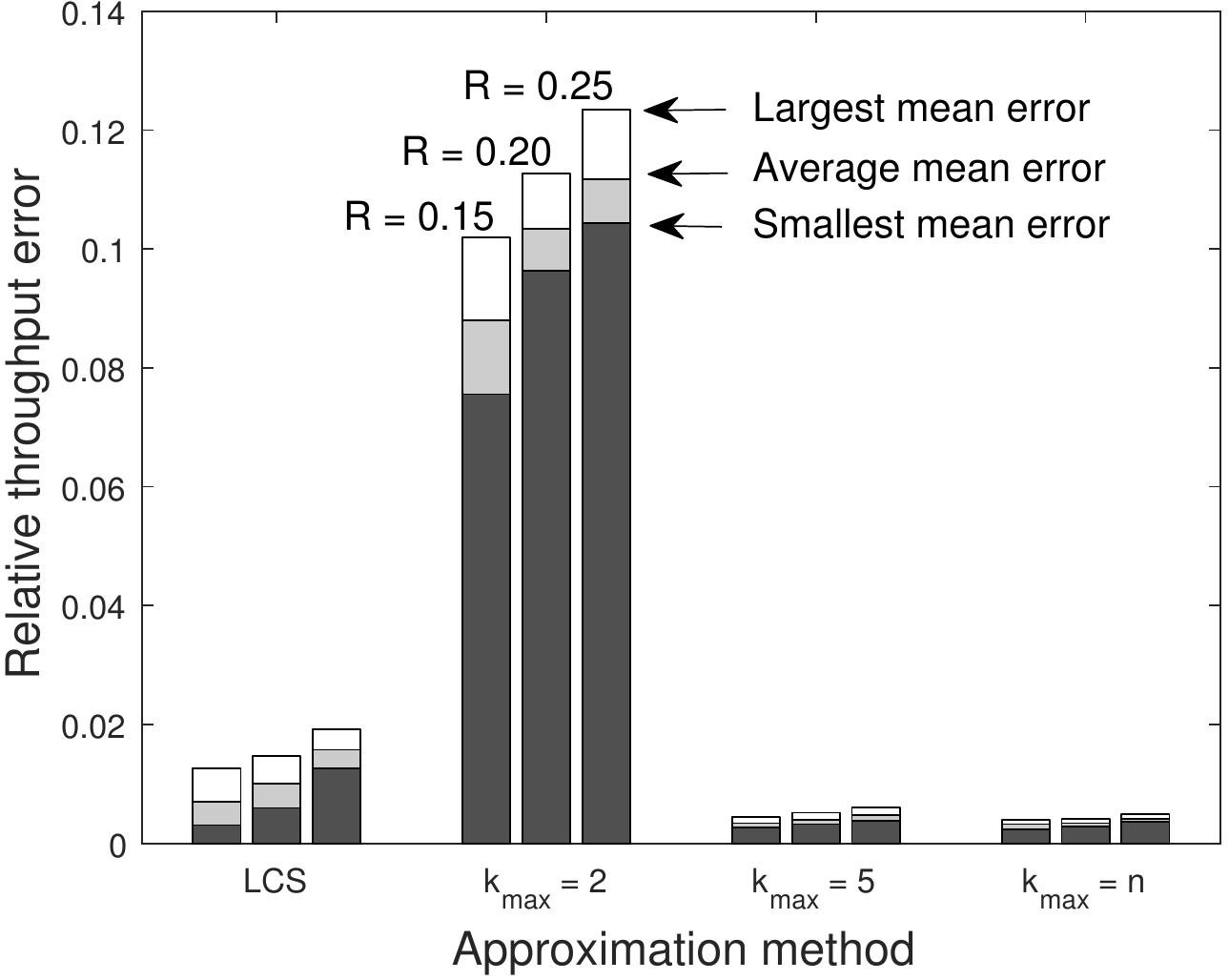}
\caption{Smallest, average and largest mean relative error in the achieved throughput 
of the LCS and size $k_{max}$ clique approximation
for $25$ conflict graphs with $n=100$ nodes when $\nu_i = 0.85/(1+d_i)$ for $R = 0.15, 0.2$ and $0.25$.}
\label{fig:relerror2}
\end{figure}
 
The results in Figure \ref{fig:relerror2} are in agreement with Figure \ref{fig:relerror}:
the LCS approximation outperforms the Bethe approximation, increasing $k_{max}$ reduces the relative throughput errors and the
LCS error increases more significantly when the graph becomes denser compared to the
size $n$ approximation. We further note that the size $k_{max} = 5$ approximation
produces errors close to the size $n$ approximation, which is a useful observation in case we
wish to limit the time needed to compute the required back-off rates. 

To get an idea on the computation times of the back-off vector for the different approximations,
we generated 1000 conflict graphs with $n=100$ and $R \in (0,0.25]$. Figure \ref{fig:times} depicts 
the average time needed to compute the vector of back-off rates for all the conflict graphs 
with a maximum clique size between $4$ and $15$ (only $9$ of the $1000$ conflict graphs contained a $15+$ clique).
The results show that the computation times of the LCS and Bethe approximation have a similar shape. They
also highlight that for denser graphs limiting $k_{max}$ may offer an attractive trade-off between the
computation times and accuracy of the approximation. 

We end by noting that the time complexity to compute the back-off rate for node $i$ only depends on
the structure of the subgraph $G_i$ induced by node $i$ and its neighbors. Thus the overall network size $n$
has no impact on the computation times and as such the approximation method is also suitable for very large
graphs as long as the size of the one hop neighborhood does not scale with the overall network size.  
Further the complexity to compute $\nu_i$ for the size $n$ approximation
is similar to performing a single iteration of the IGBP algorithm, for
which convergence to a (unique) fixed point is not guaranteed and the number of iterations grows 
with the density of the conflict graph (see \cite[Tables IV and V]{kai1}).

\begin{figure}[t]
\center
\includegraphics[width=0.45\textwidth]{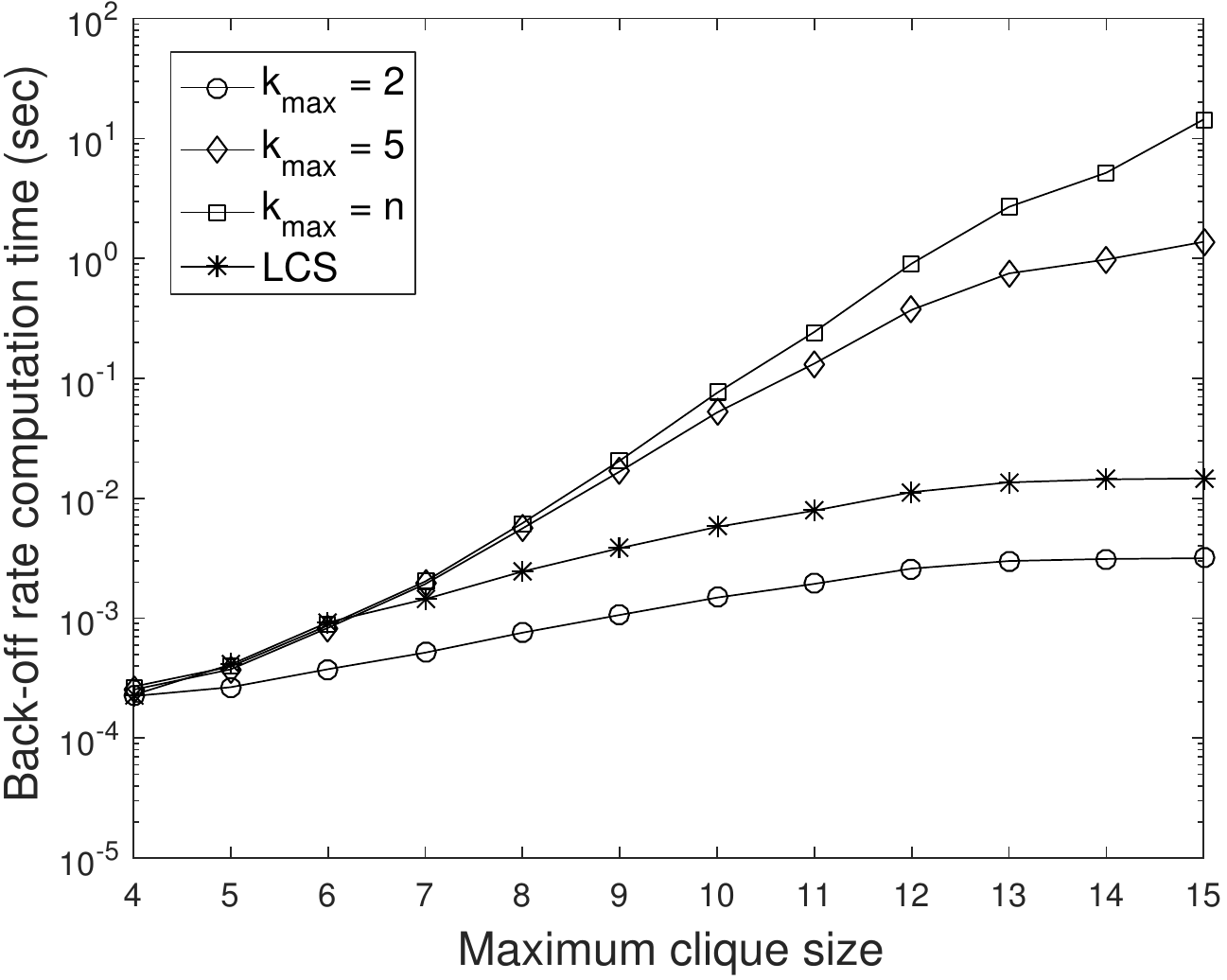}
\caption{Computation time of the vector of back-off rates as a function of the maximum clique size of $G$.}
\label{fig:times}
\end{figure}

\section{Conclusions}\label{sec:conc}
In this paper we presented the class of region-based free energy approximations for the ideal CSMA model, 
which contains the Bethe approximation of \cite{yun1} as a special case.
We obtained a closed form expression for the vector of back-off rates that corresponds to a zero gradient point
of the free energy within the set of clique beliefs (in terms of its counting numbers). 

We subsequently focused on the size $k_{max}$ clique
approximation (which can be implemented in a fully distributed manner)
and derived explicit expressions for its counting numbers as well as a recursive method
to compute the back-off rates more efficiently. We further showed that this approximation is exact on
chordal conflict graphs and coincides with a Kikuchi approximation. The latter result implies that
the size $k_{max}$ clique approximation with $k_{max}=n$ gives an explicit expression for a fixed point
of the IGBP algorithm of \cite{kai1}. 
The paper also contains an alternate proof for the back-off rates needed to achieve any achievable
throughout vector in a chordal graph $G$ presented in \cite{vanhoudt_ton17}.
 
There are a number of possible extensions to the work presented in this paper. First, while this paper has
focused on achieving a given target throughput vector, it should be possible to consider utility maximization
problems as in \cite{yun1}. Second, one could try to relax the conflict graph based interference model considered in this 
paper to a more realistic SINR (signal-to-interference-plus-noise ratio) model. In fact, such a relaxation 
of the Bethe approximation presented in \cite{yun1} was recently developed in \cite{swamy1}. 
Finally, other free energy approximation techniques such as the tree-based reparameterization framework of \cite{wainwright1}
could be considered as well.

\bibliographystyle{plain}
\bibliography{../../PhD/thesis}

\appendices

\section{Proof of Theorem \ref{th:Kikuchi}}\label{app:Kikuchi_proof}
\begin{proof}
First, as any intersection of two cliques is a clique, all the regions
of the Kikuchi approximations correspond to a region in the size $k_{max}$ clique approximation. On the other hand 
not every clique $K$ of size $k \in \{2,\ldots,k_{max}-1\}$ is necessarily a region in the Kikuchi approximation.
The proof exists in showing that the counting number $c_{R(K)} = 0$ for any region $R(K) \not\in \tilde{\mathcal{R}}$,
while $c_{R(K)} = \tilde c_{R(K)} $ otherwise. As removing regions $R$ with $c_R = 0$ does not alter the approximation, this 
suffices to prove the theorem. 

We start by noting that if the region of a clique $R(K) \not\in \tilde{\mathcal{R}}$, there exists a unique 
$i \in \{0,\ldots,s\}$ and region $S(K)\in \mathcal{R}_i$ such that
$R(K) \subset S(K)$, while $R(K) \not\subset R'$ for $R' \in \cup_{k=i}^s \mathcal{R}_k$ 
with $R' \not= S(K)$. Further, for any $R' \in \cup_{k=0}^{i-1} \mathcal{R}_k$ with $R(K) \subset R'$, we have
$S(K) \subset R'$ as $S(K)\cap R'$ (or a superset thereof) would otherwise be a region in $\mathcal{R}_{i+1}$
that contains $R(K)$. For $R(K) \not\in \tilde{\mathcal{R}}$ we define $f(K)=i$ such that the above holds.

We now prove by induction that $c_{R(K)}=0$ if $R(K) \not\in \tilde{\mathcal{R}}$ and
$c_{R(K)} = \tilde c_{R(K)}$ otherwise.
If $f(K)=0$, $S(K)$ is a maximal clique and  we have
\[c_{R(K)} = 1- c_{S(K)} - \sum_{\substack{R(K') \in \mathcal{R}':\\ K \subset K', R(K')\not\in \tilde{\mathcal{R}}}}  c_{R(K')}. \]
The latter sum can be shown to be equal to zero by noting that no such $K'$ exists when $|S(K)|=|R(K)|+1$ and 
that it is a sum of zeros when $|S(K)|>|R(K)|+1$ by applying induction on $|S(K)|-|R(K)|$. Moreover $c_{S(K)}=1$
as $S(K) \in \mathcal{R}_0$, yielding  $c_{R(K)}=0$.
Further, for any $R(K) \in \mathcal{R}_1$ 
\begin{align*}c_{R(K)} &= 1 - \hspace*{-0.4cm}\sum_{R' \in \mathcal{R}_0: R(K) \subset R'} c_{R'} -  \sum_{\substack{R(K') \in \mathcal{R}':\\ K \subset K', R(K') \not\in \tilde{\mathcal{R}}}} c_{R(K')} \\
&= \tilde c_{R(K)},\end{align*}
as for any $K'$ with $R(K')\not\in \tilde{\mathcal{R}}$ with $K \subset K'$ we have $f(K')=0$ and thus $c_{R(K')}=0$.
By induction we therefore have $c_{R(K)}=0$ for $R(K) \not\in \tilde{\mathcal{R}}$ with $f(K) < i$ and
$c_{R'} = \tilde c_{R'}$ for $R' \in \cup_{k=1}^{i} \mathcal{R}_k$. 

When $R(K) \not\in \tilde{\mathcal{R}}$ and $f(K) = i$, we can write $c_{R(K)}$ as
\begin{align}\label{eq:cRK}
c_{R(K)} &= 1 - \sum_{R' \in \tilde{\mathcal{R}}: R(K) \subset R'} c_{R'} 
- \sum_{\substack{R(K') \in \mathcal{R}':\\ K \subset K', R(K')\not\in \tilde{\mathcal{R}}}} c_{R(K')} 
\end{align} 
The former sum equals
\[
\sum_{R' \in \cup_{k=0}^{i} \mathcal{R}_k: R(K) \subset R'} c_{R'} = \tilde c_{S(K)}+\hspace*{-0.4cm}\sum_{R' \in \cup_{k=0}^{i-1} \mathcal{R}_k: R(K) \subset R'} \tilde c_{R'},
\]
where the equality $c_{R'} = \tilde c_{R'}$ and $c_{S(K)}=\tilde c_{S(K)}$ follows by induction.
Using the definition of $\tilde c_{S(K)}$ implies
\begin{align}\label{eq:firstsum}
\sum_{R' \in \cup_{k=0}^{i} \mathcal{R}_k: R(K) \subset R'} c_{R'} &=
1-\hspace*{-0.4cm}\sum_{R' \in \cup_{k=0}^{i-1}\mathcal{R}_k: S(K) \subset R'} \tilde c_{R'} \nonumber \\
&\hspace*{-2cm} + 
\sum_{R' \in \cup_{k=0}^{i-1}\mathcal{R}_k: R(K) \subset R'} \tilde c_{R'} = 1,
\end{align}
as for any $R' \in \cup_{k=0}^{i-1} \mathcal{R}_k$ with $R(K) \subset R'$ we have $S(K) \subset R'$ as noted before.

The latter sum in \eqref{eq:cRK} equals 
\[\sum_{R(K') \in \mathcal{R}': K \subset K', R(K')\not\in \tilde{\mathcal{R}}}  c_{R(K')} 1_{\{f(K')=i\}}, \]
as $f(K') \leq i$ for any $K'$ with $K \subset K'$ and by induction $c_{R(K')}=0$ if $f(K')<i$.
Hence, using \eqref{eq:cRK} and \eqref{eq:firstsum} we obtain 
\begin{align*}
c_{R(K)}= \sum_{R(K') \in \mathcal{R}': K \subset K', R(K')\not\in \tilde{\mathcal{R}}}  c_{R(K')} 1_{\{f(K')=i\}},
\end{align*} 
which is equal to zero by induction on $|S(K)|-|R(K)|$.
This also implies that $c_{R'} = \tilde c_{R'}$ for $R' \in \cup_{k=1}^{i+1} \mathcal{R}_k$.

The regions $R_{f_i}$ are clearly part of both approximations and their counting numbers
are equal to one in both cases, while the regions $R_{x_i}$ are also part of both approximations
(due to the presence of $R_{f_i}$) and their counting numbers are identical as for any clique $K$
with $i \in K$ we have $c_{R(K)} = 0$ if $R(K) \not\in \tilde{\mathcal{R}}$ and  $c_{R(K)} = \tilde c_{R(K)}$
otherwise. 
\end{proof}

\section{Proof of Theorem \ref{th:chordal}}\label{app:chordal_proof}

\begin{proof}
To establish this result we first show that the region-based free energy approximation introduced in this section
can also be obtained using the junction graph method (see \cite[Appendix A]{yedidia1}). 

A junction graph $\mathcal{G}=(V_L,V_S,E_\mathcal{G},L)$ is a directed bipartite graph
consisting of a set $V_L$ of large vertices, a set $V_S$ of small vertices, a set of directed edges from $V_L$ to $V_S$ and
a set of labels $L(v)$ for each $v \in V_L \cup  V_S$. The labels $L(v)$ are a subset of the set $I=\{x_1,\ldots,x_n\}\cup
\{f_1,\ldots,f_M\}$ of variable and factor nodes of a given factor graph. Further, for $\mathcal{G}$ to be a junction graph there are two additional 
conditions: (i) if $(v,w) \in E_\mathcal{G}$ then $L(w) \subseteq L(v)$ and (ii) for any $i \in I$
the subgraph of $\mathcal{G}$ induced by the vertices in $V_L \cup V_S$ for which $i \in L(v)$ must be a connected tree. 

Consider a region-based free energy approximation based on the regions $\mathcal{R}$ and counting numbers $c_R$.
Assume $\mathcal{R} = \mathcal{R}_L \cup \mathcal{R}_S$ such that the regions can be organized into a junction graph
and for each $R \in \mathcal{R}_L$ there is a vertex $\mathcal{V}_R \in V_L$ and
for each $R \in \mathcal{R}_S$ there is a vertex $\mathcal{V}_R \in V_S$, while $L(\mathcal{V}_R)=\mathcal{V}_R\cup \mathcal{F}_R$ for any $R \in \mathcal{R}$.
Further assume that $c_R = 1$ for $R \in \mathcal{R}_L$ and $c_R = 1 - d_R$ for $R \in \mathcal{R}_S$, where
$d_R$ is the number of neighbors of $R$ in the junction graph. If the believes match the exact marginal probabilities
and the junction graph is a tree, the region-based entropy is exact, i.e., matches the Gibbs entropy, and
\begin{align}\label{eq:pxgen}
p(x) = \frac{\prod_{R \in \mathcal{R}_L} p_{R}(x_{R})}{\prod_{R \in \mathcal{R}_S} p_{R}(x_{R})^{d_R-1}}, 
\end{align}
with $p_R(x_R) = \sum_{x_j\not\in \mathcal{V}_R} p(x_1,\ldots,x_n)$, as argued in \cite[Appendix A]{yedidia1}. 

We now organize the $2n+2|\mathcal{K}_G|-1$ regions of our chordal region-based approximation into a junction graph $\mathcal{G}$ 
by letting $V_L = \{v_K | K \in \mathcal{K}_G\} \cup \{v_{f_i} | i \in V\}$,
$V_S = \{v_{(K,K')} | (K,K') \in \mathcal{E} \}\cup \{v_{x_i} | i \in V\}$. For each edge $(K,K') \in \mathcal{E}$
we add an edge from $v_K$ to $v_{(K,K')}$ and one from $v_{K'}$ to $v_{(K,K')}$, while for each $v_{x_i}$ we
add an edge from $v_{f_i}$ to $v_{x_i}$ and one from $v_{K_i}$ to $v_{x_i}$, where $K_i$ is a randomly selected
$K \in \mathcal{K}_G$ such that $i \in K$. Note that all the nodes $v \in V_S$ have exactly two neighbors, that is, $d_R = 2$
for $R \in \mathcal{R}_S$. The labels are defined by $L(\mathcal{V}_R)=\mathcal{V}_R\cup \mathcal{F}_R$ for all the regions $R \in \mathcal{R}$. 

It is trivial to check that condition (i) for $\mathcal{G}$ to be a junction graph holds. Condition (ii) follows by
noting that for any $i \in V$ the subgraph of the clique tree $T$ induced by the cliques $K \in \mathcal{K}_G$ with 
$i\in K$ is a subtree of $T$ and the same holds for any $i,j \in V$ if we look at the subgraph induced by the maximal cliques containing
both $i$ and $j$. As $T$ is a tree, so is the junction graph $\mathcal{G}$, which implies that the entropy is exact and 
\eqref{eq:px} follows from \eqref{eq:pxgen} by noting that $d_R = 2$ and $p_{R_{x_i}}(x_i) = p_{R_{f_i}}(x_i)$ for $i \in V$.

The expression for $H(p)$ follows from the well-known fact that the Gibbs free energy $F(b) = U(b) - H(b)$
is minimized when $b(x) = p(x)$ for all $x$, where it attains the value $-\ln(Z)$, thus due to \eqref{eq:URbR2} 
we have $-\ln(Z) = - \sum_{i \in V} p_i(1) \ln(\nu_i) - H(p)$.    
\end{proof}

\section{Proof of Theorem \ref{th:coin}}\label{app:coin}

\begin{proof}
The proof exists in showing that for any region $R \in \cup_{k=1}^s \mathcal{R}_k$ the counting number $\tilde c_R$ is given by
\[ \tilde c_R =  -|\{(K,K') \in \mathcal{E} | \mathcal{V}_R = V_{R_{(K,K')}}\}|-1_{\{|\mathcal{V}_R|=1\}}.\] 
In other words, if the region $R$ corresponds to a clique $S$ of size $2$ or more, $-\tilde c_R$ is equal to
the number of edges $(K,K')$ in the clique tree $T$ such that $S = K \cap K'$, while for
the regions with $\mathcal{V}_R = \{x_i\}$ it is the same number plus one (due to the region $R_{f_i}$).   

We now proceed by induction on $i$. Assume $R \in \mathcal{R}_1$ and its corresponding clique $S$ has size $2$ or more and
is a subset of $b$ maximal cliques $\{K_1,\ldots,K_b\}$, then $\tilde c_R = 1 - b$. Further, the subtree of $\mathcal{T}$ induced by 
$\{K_1,\ldots,K_b\}$ contains exactly $b-1$ edges $(K,K')$ and $S \subseteq K \cap K'$. We now argue that 
$S = K \cap K'$ for each of these $b-1$ edges $(K,K')$. Assume one of these edges $(K,K')$ is such that $S \subset K \cap K'$,
meaning the intersection of $K$ and $K'$, two maximal cliques, contains $S$ as a strict subset. Then, $R$ cannot be
part of $\mathcal{R}_1$ as any region that was obtained as an intersection between two regions of $\mathcal{R}_0$,
but that is part of a larger intersection of two such regions was removed from $\mathcal{R}_1$ during its construction.

Next assume $R \in \mathcal{R}_{i+1}$ and its corresponding clique $S$ is of size $2$ or more, then
\begin{align*}
\tilde c_R &= 1 - \hspace*{-0.4cm}\sum_{R' \in \cup_{k=0}^i \mathcal{R}_k: R \subset R'} \tilde c_R 
= 1 - |\{K \in \mathcal{K}_G | S \subset K\}|\\
& + \sum_{k=1}^i \sum_{R' \in \mathcal{R}_k: R \subset R'}
|\{(K,K') \in \mathcal{E} | \mathcal{V}_{R'} = V_{R_{(K,K')}}\}|, 
\end{align*}
by induction. Hence $-\tilde c_R$ thus matches the number of edges in the subtree $T_S$ of $T$
induced by the maximal cliques that contain $S$ minus the number of edges $(K,K')$ for which $K \cap K'$
corresponds to a region in $\cup_{k=0}^i \mathcal{R}_k$. Note if $(K,K') \in T$ and $K \cap K'$ corresponds to 
a region in  $\cup_{k=0}^i \mathcal{R}_k$ with $S \subseteq K \cap K'$ then  $S \not= K \cap K'$ as $R \in \mathcal{R}_{i+1}$.
The claim  for $\tilde c_R$ therefore follows
if the edges of this subtree for which $K\cap K'$ does not correspond to a region in $\cup_{k=0}^i \mathcal{R}_k$ are
such that $K \cap K' = S$.  This is the case as any edge in the subtree $T_S$ with an intersection that is a superset of $S$
and that is not part of $\cup_{k=0}^i \mathcal{R}_k$ would cause the removal of $R$ from $\mathcal{R}_{i+1}$ during
its construction.

For a region $R$ with $\mathcal{V}_R = \{x_i\}$ and $\mathcal{F}_R = \emptyset$ the same reasoning applies except that $R$
is also a subset of $R_{f_i}$ and $\tilde c_{R_{f_i}}=1$ causing the minus one in $\tilde c_R$. 
\end{proof}

\end{document}